\documentclass{jca}
\usepackage{amsmath, amsthm, pstricks, pst-node}

\newtheorem{thm}{Theorem}
\newtheorem{claim}[thm]{Claim}
\newtheorem{prop}[thm]{Proposition}
\newtheorem{defn}[thm]{Definition}
\newtheorem{lemma}[thm]{Lemma}

\newcommand{\comment}[1]{}
\def\A{\mathcal{A}}
\def\tmo{\mathtt{the.member}}

\begin{document}

\title{Efficient Divide-and-Conquer Implementations Of Symmetric FSAs}

\author{David A.\ G.\ Pritchard\email{daveagp@gmail.com}}

\institute{Department of Combinatorics \& Optimization, University
of Waterloo, Canada}

\maketitle

\begin{abstract}
A deterministic finite-state automaton (FSA) is an abstract sequential machine
that reads the symbols comprising an input word one at a time. An
FSA is \emph{symmetric} if its output is independent of the order in
which the input symbols are read, i.e., if the output is invariant
under permutations of the input. We show how to convert a symmetric
FSA $\A$ into an automaton-like divide-and-conquer process whose
intermediate results are no larger than the size of $\A$'s memory.
In comparison, a similar result for general FSA's has been long
known via functional composition, but entails an exponential
increase in memory size. The new result has applications to parallel
processing and symmetric FSA networks.
\end{abstract}


\section{Introduction}
One of the simplest models of computation is the \emph{deterministic finite state
automaton} (FSA). Although FSAs are often considered to act as
solitary computing devices (e.g., in the classical string matching
algorithm of Knuth, Morris, and Pratt \cite{KMP77}) they can also be
connected together to form a computing network (e.g., in cellular
automata and the models of \cite{AA+06, Milgram75}).

A \emph{symmetric} automaton is one that will produce the
same output even if its inputs are permuted. Symmetric FSAs are natural building blocks for fault-tolerant computation networks. In previous work with Vempala~\cite{PV06} we showed there are symmetric FSAs implementing fault-tolerant randomized algorithms for the following tasks: implicit
2-approximate census (via an algorithm due to
Milgram~\cite{Milgram75}), network search (via breadth-first search
and greedy traversal), basic connectivity problems (finding all
bridges), and leader election.

Precisely, the \emph{finite-state symmetric graph automaton} (FSSGA)
model introduced in \cite{PV06} is that a copy of the same symmetric
FSA is placed at every node of a graph (network); when a node
``activates" to advance its state, it obtains one input symbol from
each neighbour without regard to order. In sum, FSSGA are like
symmetric cellular automata but generalized in the sense that the underlying
graph does not have to be regular.
Three models of
symmetric automata are given in \cite{PV06}: in the
\emph{sequential} model each node is a (sequential) FSA, in the
\emph{parallel} model each node uses divide-and-conquer on its
inputs (in a way that will be defined precisely later), and in the \emph{mod-thresh} model each node applies a
finite-size formula (analogous to a regular expression) to update
its state. One of the main results of \cite{PV06} is that these
three models are equivalent; e.g., for any symmetric FSA there
exists a divide-and-conquer process to compute the same function.
Unfortunately, for the particular construction given in \cite{PV06},
an exponential increase in the size of the state space is required.

More generally, an efficient way to simulate \emph{any} FSA with
divide-and-conquer has been known for decades. The basic technique
is sometimes called \emph{functional composition} as applied to
\emph{parallel prefix}. Ladner and Fischer used the technique in
1977 \cite{LF80} on the PRAM model of parallel computing; see also
\cite{MF95} for an implementation in mesh networks. The basic idea
is that for any single character $\sigma$, the transition of the FSA
on that character can be viewed as function $f_{\sigma}$ from the
FSA's state space back to itself, and the computation of the FSA on
a string $w = w_1w_2\dotsb w_k$ is essentially determined by the
composition of functions $f_w := f_{w_k}\circ\dotsb\circ f_{w_2}
\circ f_{w_1}.$ In turn, this composition problem lends itself
easily to divide-and-conquer: break the string into two parts $w =
uv$, compute the compositions $f_u$ and $f_v$ for the two parts, and
return $f_v \circ f_u$. Like the transformation of \cite{PV06} for
symmetric automata, the size of intermediate results increases
exponentially, since for a state space $Q$ there are $|Q|^{|Q|}$
functions from $Q$ to $Q$.

The main contribution of this paper is that for a \emph{symmetric}
FSA, no increase in the state space size is
necessary. We present the result (Theorem \ref{thm:main})
after introducing our
notation. The resulting small-state-space divide-and-conquer process
is applicable to the
PRAM setting, so e.g.\ for symmetric FSAs we are able to decrease the
working memory used by the divide-and-conquer approaches of \cite{LF80, MF95}.
For high-degree FSSGAs and the special case of symmetric cellular automata,
divide-and-conquer is a natural way for each node to read its neighbours'
states, as we will illustrate in Section \ref{sec:prelim}; our main
result permits such divide-and-conquer processes to be more
memory-efficient.

\section{Preliminaries}\label{sec:prelim}
We denote an FSA by the tuple $(\Sigma, Q, q_0, \{f_\sigma\}_{\sigma
\in \Sigma}, O, \beta)$. Here $\Sigma$ is a finite set called the
\emph{input alphabet}, $Q$ is a finite set called the \emph{state
space}, $q_0$ is an element of $Q$ called the \emph{initial state},
each $f_\sigma$ is a function from $Q$ to $Q$ called the
\emph{transition function of $\sigma$}, $O$ is a finite \emph{output
set}, and $\beta$ is an \emph{output function} from $Q$ to $O$.

\begin{defn}[FSA] An \emph{FSA} is any tuple $\A = (\Sigma, Q, q_0, \{f_\sigma\}_{\sigma \in \Sigma}, O, \beta)$ of the form described above.
\end{defn}

Let $\Sigma^*$ denote the set of all strings over $\Sigma$, and let
$f \circ g$ denote the functional composition of $f$ and $g$,
defined by $(f \circ g)(x) = f(g(x)).$ It is convenient to extend
the definition of $f$ to strings via functional composition. Namely,
for a string $w = w_1w_2\dotsb w_k,$ define
$$f_w := f_{w_k} \circ f_{w_{k-1}} \circ \dotsb \circ f_{w_2} \circ
f_{w_1},$$ and by convention, where $\lambda$ denotes the empty
string, let $f_\lambda$ be the identity function on $Q$. In
particular, we obtain the identity $f_{uv}(q) = f_v( f_u(q))$ for
any strings $u, v \in \Sigma^*$ and any $q \in Q$. Let $\Sigma^+$
denote the set of nonempty strings over $\Sigma$; the empty string
is excluded to agree with the divide-and-conquer model later on. Our
definition of $f_w$ affords a concise description of computation for
an FSA.

\begin{defn}[FSA computation] An FSA $\A = (\Sigma, Q, q_0, \{f_\sigma\}_{\sigma \in \Sigma}, O, \beta)$ \emph{computes} the function $\nu_\A : \Sigma^+ \to O$ defined by
$$\nu_\A(w) := \beta(f_w(q_0)).$$
\end{defn}

\comment{The FSA is said to \emph{compute} a function from
$\Sigma^*$ to $O$, where $\Sigma^*$ is the set of all strings over
$\Sigma$. Namely, on input $w_1w_2\dotsb w_k \in \Sigma^*$, the
output of this function is defined to be
$$\beta(f_{w_k}(f_{w_{k-1}}(\dotsb(f_{w_2}(f_{w_1}(q_0)))))).$$}
Note that the traditional model where the FSA accepts
or rejects strings depending on the final  state can be modeled by
setting $O = \{accept, reject\}$ and defining $\beta(q)=accept$ iff
$q$ is an accepting state. We use the multi-output version because
it is more natural in some settings, e.g., the FSSGA model.

We represent a divide-and-conquer automaton by a tuple $(\Sigma, Q,
\alpha, c, O, \beta)$. As before $\Sigma$ is the input alphabet, $Q$
is the state space, $O$ is the output set and $\beta$ is the output
function. Here $\alpha$ is an \emph{input function} from $\Sigma$ to
$Q$ and $c$ is a \emph{combining function} from $Q \times Q$ to $Q$.
Informally, the divide-and-conquer automaton runs according to the
following rules:
\begin{enumerate} \item apply $\alpha$ to all input characters
\item combine states arbitrarily using $c$ until only one
 state $q^*$ is left \item output
$\beta(q^*)$.\end{enumerate}
Our definition will require that the end result of the computation
is the same no matter how the arbitrary choices of combination are made.

To give our formal definition, we use a set-valued function $\chi$
that maps each nonempty string to a subset of $Q$ so that $q^* \in
\chi(w)$ iff, dividing inputs arbitrarily, the input $w$ could
produce $q^*$ as the final  state. We denote the length of $w$ by
$|w|$.
\begin{defn}[DCA] Let $\mathcal{A}'$ denote the tuple $(\Sigma, Q, \alpha, c, O, \beta)$
as described above. Define $\chi_{\A'}(w)$ for $w \in \Sigma^+$
recursively as follows: if $|w|=1$, say $w$ consists of the
character $\sigma$, then $\chi_{\A'}(w) := \{\alpha(\sigma)\};$
otherwise (for $|w| \geq 2$)
\begin{equation}
\chi_{\A'}(w) := \bigcup_{\substack{(u, v) : uv=w}}\{c(q^*_u, q^*_v)
\mid q^*_u \in \chi_{\A'}(u), q^*_v \in \chi_{\A'}(v)\}
\label{eq:dca}
\end{equation}
where $(u, v)$ ranges over all partitions of $w$ into two nonempty
substrings. We say that $\A'$ is a \emph{divide-and-conquer
automaton (DCA)} if for all $w \in \Sigma^+$,
\begin{equation}
|\{\beta(q^*) \mid q^* \in \chi_{\A'}(w)\}|=1. \label{eq:kobe}
\end{equation}
\end{defn}

The previous definition amounts to saying that the output of a
divide-and-conquer automaton should be well-defined regardless of
how the dividing is performed. For a singleton set $S$ let $\tmo(S)$
be a function that returns the element of $S$, i.e., it ``unwraps"
the set.

\begin{defn}[DCA computation]A DCA $\A' = (\Sigma, Q, \alpha, c, O, \beta)$ \emph{computes}
the function $\nu_{\A'} : \Sigma^+ \to O$ defined by
\begin{equation}
\nu_{\A'}(w) = \tmo(\{\beta(q) \mid q \in \chi_{\A'}(w)\}).\label{eq:qe}
\end{equation}
\end{defn}

Figure 1 illustrates how a node in an FSA-based computing network
could make use of the divide-and-conquer methodology. Specifically,
when reading the states of all neighbours
the node can process and combine inputs from
its neighbours in parallel rather than one-by-one. As a function of
the neighbourhood size $|\Gamma|$ (i.e.\ the degree)
the circuit depicted has depth $O(\log
|\Gamma|)$ and hence this approach would lead to efficient
 physical implementation for large neighbourhoods.

\begin{figure}[htb]
\begin{center} \leavevmode
\begin{pspicture}(-4,-4)(4,4)
\psset{unit=0.75} {\psset{doubleline=true} \cnodeput(0,0){v}{$v$} \cnodeput(-4,4){n1}{$n$}
\cnodeput(0,4){n2}{$n$} \cnodeput(4,4){n3}{$n$}
\cnodeput(-4,0){n4}{$n$} \cnodeput(4,0){n5}{$n$}
\cnodeput(-4,-4){n6}{$n$} \cnodeput(0,-4){n7}{$n$}
\cnodeput(4,-4){n8}{$n$}} \psset{arrows=->}
\rput(3,3){\rnode{A3}{\psframebox{$\alpha$}}} \ncline{n3}{A3}
\rput(1,3){\rnode{A2}{\psframebox{$\alpha$}}} \ncline{n2}{A2}
\rput(2,2){\rnode{C23}{\psframebox{$c$}}} \ncline{A2}{C23}
\ncline{A3}{C23} \rput(-3,3){\rnode{A1}{\psframebox{$\alpha$}}}
\ncline{n1}{A1} \rput(-3,1){\rnode{A4}{\psframebox{$\alpha$}}}
\ncline{n4}{A4} \rput(-2,2){\rnode{C14}{\psframebox{$c$}}}
\ncline{A1}{C14} \ncline{A4}{C14}
\rput(0,2){\rnode{Ct}{\psframebox{$c$}}} \ncline{C14}{Ct}
\ncline{C23}{Ct} \rput(-2,0){\rnode{C}{\psframebox{$c$}}}
\ncline{Ct}{C} \rput(-1,0){\rnode{B}{\psframebox{$\beta$}}}
\ncline{C}{B} \ncline{B}{v}
\rput(3,-1){\rnode{A5}{\psframebox{$\alpha$}}}
\rput(-3,-3){\rnode{A6}{\psframebox{$\alpha$}}}
\rput(-1,-3){\rnode{A7}{\psframebox{$\alpha$}}}
\rput(3,-3){\rnode{A8}{\psframebox{$\alpha$}}}
\rput(-2,-2){\rnode{C67}{\psframebox{$c$}}}
\rput(2,-2){\rnode{C58}{\psframebox{$c$}}}
\rput(0,-2){\rnode{Cb}{\psframebox{$c$}}} \ncline{n5}{A5}
\ncline{n6}{A6}\ncline{n7}{A7}\ncline{n8}{A8} \ncline{A5}{C58}
\ncline{A8}{C58} \ncline{A6}{C67} \ncline{A7}{C67} \ncline{C67}{Cb}
\ncline{C58}{Cb} \ncline{Cb}{C}
\end{pspicture}
\end{center}
\caption{An FSA in a network updates its state via
divide-and-conquer. The node $v$ is activating and its neighbours
are labeled $n$. The lines carry values from tail to head, and the
boxes apply functions, like in a circuit diagram. Each neighbour
supplies an input symbol and the divide-and-conquer process produces
an output symbol which is used by $v$ to update its state.}
\label{fig:example}
\end{figure}

We denote by $Q^Q$ the set of all functions from $Q$ to $Q$. We mentioned the following well-known (e.g., \cite{LF80}) result
earlier:
\begin{thm}\label{thm:compose}
Given any FSA $\A$, there is a DCA $\A'$ such that $\nu_\A =
\nu_{\A'},$ i.e., $\A$ and $\A'$ compute the same function.
\end{thm}
\begin{proof}
Define $\A' = (\Sigma, Q^Q, \sigma \mapsto f_\sigma, (f_1, f_2) \mapsto f_2
\circ f_1, O, \beta)$.
\end{proof}

Conversely, as was observed in \cite{PV06}, any divide-and-conquer
automaton can be easily rewritten in sequential form since a
sequential FSA can be thought of as conquering one input at a time.

The particular result we want to prove pertains only to symmetric
automata, which we now define formally.

\begin{defn} Suppose that $\A$ is an FSA or a DCA.
We say that $\A$ is \emph{symmetric} if for every $w \in \Sigma^+$
and every permutation $w'$ of $w,$ $\nu_\A(w)=\nu_\A(w').$
\end{defn}

The main result of the present paper is the following, which is a
more efficient version of Theorem \ref{thm:compose} for symmetric
FSA's.
\begin{thm}\label{thm:main}
Given any symmetric FSA $\A = (\Sigma, Q, q_0, f, O, \beta)$, there
is a DCA $\A' = (\Sigma, Q', \alpha, c, O, \beta')$ such that
$\nu_\A=\nu_{\A'}$  \emph{and $|Q'|\leq|Q|$}.
\end{thm}

In the next section, we prove a supporting lemma for later use. In
Section \ref{sec:mainproof} we complete the proof of Theorem
\ref{thm:main}. In Section \ref{sec:ideas} we mention some ideas for
future investigation.

\section{Looking Inside A Symmetric FSA}\label{sec:inside}
The key to Theorem \ref{thm:main} is to focus on
automata with specific irredundany properties. Symmetry of an
automaton is a \emph{black-box property} --- the definition only
cares about the correspondence of final outputs when the inputs are
permutations of one another, regardless of the internal structure of the automaton. We now describe how this black-box property (symmetry),
when combined with irredundancy requirements, implies
a structural property --- namely, that
the transition functions must commute.

\begin{defn}[\cite{HU79}] Let $\A = (\Sigma, Q, q_0, f, O, \beta)$ be an FSA
and let $q \in Q$. The state $q$ is said to be \emph{accessible}
if for some string $w \in \Sigma^*,$ $f_w(q_0) = q$. We say $\A$ is \emph{accessible}
if every state in $Q$ is accessible.
\end{defn}

\begin{defn}[\cite{HU79}] Let $\A = (\Sigma, Q, q_0, f, O, \beta)$ be an FSA
and let $q, q' \in Q$. The states $q$ and $q'$ are said to be
\emph{distinguishable} if for some string $w \in \Sigma^*,$
$\beta(f_w(q)) \neq \beta(f_w(q'))$. We say $\A$ is \emph{distinguishable} if every
pair of states in $Q$ is distinguishable.
\end{defn}

As we will later make precise, every FSA can be rewritten in an accessible, distinguishable way.
This gives some general applicability to the following lemma.

\begin{lemma}[Commutativity Lemma] \label{lemma:ma}Let $\A = (\Sigma, Q, q_0, f, O, \beta)$
be a symmetric FSA that is accessible and distinguishable. Then the functions
$\{f_\sigma\}_{\sigma\in\Sigma}$ commute.
\end{lemma}

We defer the proof of the lemma to the end of this section. In order to explain how it is useful,
we recall the following additional definitions.

\begin{defn}[\cite{HU79}] Two automata $\A, \A'$ are
\emph{equivalent} if they compute
the same function, i.e.\ if $\nu_{\A} = \nu_{\A'}$. An FSA $\A$ is \emph{minimal} if for every
FSA $\A'$ equivalent to $\A$, $\A'$ has at least as many states as $\A$.\end{defn}

It is not hard to see that any minimal FSA must be accessible (or else we could remove some
states) and distinguishable (or else we could collapse some states)\footnote{Interestingly,
the converse is also true: any accessible, distinguishable FSA is minimal. See \cite{HU79} for
a derivation of this result as a corollary of the Myhill-Nerode theorem;
adapting the proof from accept/reject automata to our more general model
is straightforward.}. It is also not hard
to see that for every FSA $\A$ there exists a minimal equivalent FSA $\A'$; such minimization
can be performed algorithmically in $poly(|Q|,|\Sigma|)$ time, e.g.\ using
an approach of Hopcroft \cite{Hop71}.
In sum, for any FSA we can efficiently obtain an equivalent FSA meeting the conditions of
Lemma \ref{lemma:ma}, which we now prove.

\begin{proof}[Proof of Lemma \ref{lemma:ma}]
Suppose for the sake of contradiction that not all of the functions
$f$ commute. Then $f_{\sigma_1} (f_{\sigma_2} (q)) \neq f_{\sigma_2}
(f_{\sigma_1} (q))$ for some $\sigma_1, \sigma_2 \in \Sigma, q \in
Q.$ We want to show that this discrepancy can be ``continued" to a
violation of symmetry. Let $q_1$ denote $f_{\sigma_2} (f_{\sigma_1}
(q))$ and $q_2$ denote $f_{\sigma_1} (f_{\sigma_2} (q)).$

First, since $q$ is accessible, there exists some string
$w_\ell$ such that $f_{w_\ell}(q_0) = q.$ Second, since $q_1$ and
$q_2$ are distinguishable, there exists some string $w_r$ such that
$\beta(f_{w_r}(q_1)) \neq \beta(f_{w_r}(q_2))$. Now putting things
together we have
$$\beta(f_{w_\ell\sigma_1\sigma_2w_r}(q_0)) = \beta(f_{\sigma_1\sigma_2w_r}(q)) = \beta(f_{w_r}(q_1)).$$
Similarly
$$\beta(f_{w_\ell\sigma_2\sigma_1w_r}(q_0)) = \beta(f_{w_r}(q_2)) \neq \beta(f_{w_r}(q_1)).$$
Hence $\A$ outputs different values under the inputs
$w_\ell\sigma_1\sigma_2w_r$ and $w_\ell\sigma_2\sigma_1w_r$; since
these inputs are permutations of one another, this means $\A$ is not
symmetric.
\end{proof}


\section{Proof of Theorem 7}\label{sec:mainproof}
We are given that $\A = (\Sigma, Q, q_0, f, O, \beta)$ is a
symmetric FSA and without loss of generality it is minimal.
For each $q \in Q,$ define $r[q] \in
\Sigma^*$ to be a fixed \emph{representative string} that generates
state $q$ from $q_0$, i.e., such that
$$f_{r[q]}(q_0)=q$$ holds. Each $r[q]$ is guaranteed to exist since $q$
is accessible. These $r[q]$ remain fixed for the remainder of
the proof.

We need the following claim, which roughly says that every string
$w$ is interchangeable with the representative string $r[f_w(q_0)]$.
We know they are interchangeable when they are read first, but using
the commutativity of the $f$'s, we can show they are interchangeable
when read later.
\begin{claim} \label{claim:key} For each $w \in \Sigma^*$ we have $f_w = f_{r[f_w(q_0)]}.$
\end{claim}
\begin{proof}
For any $q \in Q,$ alternately applying the definition of $r[\cdot]$
and the commutativity of the $f$'s, we have
\begin{align*} f_w(q) &= f_{w}  (f_{r[q]}(q_0)) = f_{r[q]} (f_{w}
(q_0)) \\&= f_{r[q]} (f_{r[f_w (q_0)]} (q_0)) = f_{r[f_w(q_0)]}
(f_{r[q]} (q_0)) = f_{r[f_w (q_0)]} (q). \tag*{\qedhere}
\end{align*}\end{proof}

\subsection{The Construction}
Here we define the divide-and-conquer automaton $\A' = (\Sigma, Q',
\alpha, c, O, \beta')$. Namely, let $Q' = Q$, $\beta' = \beta$,
define $\alpha(\sigma) := f_\sigma(q_0)$ and define $c(q, q') :=
f_{r[q']}(q).$ It remains to prove that the construction is correct,
i.e., that $\nu_\A = \nu_{\A'}.$ Our recursive proof uses the idea
outlined previously, that each string $w$ is essentially
interchangeable with $r[f_w(q_0)].$

\begin{claim}\label{claim:final}
For any nonempty string $w \in \Sigma^+,$ the set $\chi_{\A'}(w)$ is
a singleton and $\tmo(\chi_{\A'}(w)) = f_w(q_0)$.
\end{claim}
\begin{proof}
We proceed by induction on $|w|$.
\begin{description}
\item{Base case:} If $w$ has length 1, say it consists of the character $\sigma$, then $f_w(q_0) = f_\sigma(q_0)$, and by the definition of $\chi$, we have $\chi_{\A'}(w) =  \{\alpha(\sigma)\} = \{f_\sigma(q_0)\}.$ Thus the claim is satisfied.
\item{Inductive step:} Now $w$ has length 2 or more. The induction statement to be proved is $\chi_{\A'}(w) = \{f_w(q_0)\}.$ Recalling Equation \eqref{eq:dca}, which defines $\chi$ in this case, this is equivalent to saying that
    \begin{align}
    \mbox{for all partitions $w = uv$ of $w$ into two nonempty substrings,} \notag \\ c(\tmo(\chi_{\A'}(u)), \tmo(\chi_{\A'}(v))) = f_w(q_0). \label{eq:foo}
    \end{align}
    By the induction hypothesis, the left-hand side of \eqref{eq:foo} is equal to
    \begin{equation}c(f_u(q_0), f_v(q_0)).\label{eq:foor}\end{equation}
    Applying the definition of $c$, we find that the value \eqref{eq:foor} is in turn equal to
    $f_{r[f_v(q_0)]}(f_u(q_0)).$ Finally, applying Claim \ref{claim:key} we see that
    the value \eqref{eq:foor} is equal to $f_v(f_u(q_0)) = f_w(q_0)$, as desired. \hfill \qedhere
\end{description}
\end{proof}

\begin{proof}[Proof of Theorem \ref{thm:main}]
As outlined previously, minimizing $\A$ makes it accessible and distinguishable,
without changing $\nu_\A$. Now consider
the DCA $\A'$ as defined previously. On any input $w \in \Sigma^+$,
using Claim \ref{claim:final},
$$\nu_{\A'}(w) = \beta(\tmo(\chi_{\A'}(w))) = \beta(f_w(q_0)) = \nu_{\A}(w). $$
Hence $\A$ and $\A'$ compute the same function (i.e., they are equivalent).

Since the state space of $\A'$ is $Q$, and since $Q$ could only have
gotten smaller when $\A$ was minimized,
the state space of the DCA $\A'$ is indeed no larger than the state space of
the original FSA.
\end{proof}

One might question whether any result similar to
Theorem \ref{thm:main} is possible if we discard the symmetry
requirement. The following result gives a negative answer to this question
and shows that the exponential state space increase
of Theorem \ref{thm:compose} is best possible.
\begin{prop}For any integer $n \geq 1$, there is an $n$-state FSA
$\mathcal{A}$ on a three-symbol alphabet $\Sigma$
so that any DCA equivalent to $\A$ has at
least $n^n$ states.
\end{prop}
\begin{proof}
Let $Q$ be a set of $n$ states and $\Sigma$ a set of size 3.
D\'{e}nes \cite{Denes66} showed that $Q^Q$,
viewed as a semigroup under the operation of composition, has a generating
set of size 3. We choose $\{f_\sigma\}_{\sigma \in \Sigma}$ to be this
generating set; this implies that for every function $g:Q \to Q$, there
is a string $w[g] \in \Sigma^*$ so that $f_{w[g]} = g$. We define
$O=Q$, $\beta$ to be the identity function, and we choose $q_0 \in Q$
arbitrarily; this completes the definition of the FSA $\mathcal{A}$.

Suppose for the sake of contradiction that there exists a DCA
$\mathcal{A}'$ that computes $\nu_\mathcal{A}$, and that this DCA's state
space $Q'$ has $|Q'| < n^n$. By the pigeonhole principle
there are two distinct functions $g_1, g_2 \in Q^Q$ so that
$\chi_{\mathcal{A}'}(w[g_1]) \cap \chi_{\mathcal{A}'}(w[g_2]) \neq
\emptyset$,
since each $\chi_{\mathcal{A}'}(\cdot)$ is a nonempty subset
of $Q'$. Let $\hat{q} \in Q$ denote a state for which
$g_1(\hat{q}) \neq g_2(\hat{q})$ and let $q' \in Q'$ denote any element
of $\bigcap_{i=1,2} \chi_{\mathcal{A}'}(w[g_i])$.

Now let $h : Q \to Q$ be any function for which $h(q_0)=\hat{q}$. We claim
that the two input strings $w[h]w[g_i]$ for $i=1,2$ produce
different outputs
under $\mathcal{A}$ and the same output under $\mathcal{A}'$, providing
the contradiction. To see that the outputs under $\A$ are different,
observe that
$$\nu_\A(w[h]w[g_i]) = \beta(f_{w[h]w[g_i]}(q_0)) = \beta(g_i(h(q_0))) =
g_i(\hat{q})$$
and since $g_1(\hat{q}) \neq g_2(\hat{q})$, we are done. To see that
the outputs under $\A'$ are the same,
let $\overline{q}'$ denote any element of $\chi_{\mathcal{A}'}(w[h])$
and notice that $c(\overline{q}', q') \in \chi_{\mathcal{A}'}(w[h]w[g_i])$
for $i=1,2$; then recalling Equations \eqref{eq:kobe} and \eqref{eq:qe},
we see that
$\nu_{\mathcal{A}'}(w[h]w[g_1]) = \nu_{\mathcal{A}'}(w[h]w[g_2])$ as
claimed.
\end{proof}

\section{Extensions}\label{sec:ideas}
We mention in this sections some extensions of FSAs and ask if
analogues of Theorem \ref{thm:main} hold for them. Some of these
issues were raised previously in \cite{PV06}.

First, the main result of this paper is not suitable in the following
natural situation.
Suppose
the input alphabet and state space are both the set of all $k$-bit
binary strings, i.e.\ $\Sigma = Q = \{0,1\}^k$, and that
the transition function
$f_\sigma(q)$ is some polynomial-time Turing-computable function of
$\sigma$ and $q$ (and similarly for $\beta$). For such an FSA,
$\nu_\A(w)$ can be computed in $|w|\cdot poly(k)$
time. If $\A$ is symmetric we can simulate it by a DCA using
Theorem \ref{thm:main} but this approach takes exponential
time in $k$, since minimizing $\A$ requires looking at all of its
$2^k$ states. Functional composition (Theorem \ref{thm:compose})
has the same issue. Thus, the open problem is to determine if a $poly(k)$-time
 technique exists
to convert a symmetric FSA of this type into a DCA.

Second, a variant of the above model might allow the string lengths
to grow as some function $k(m)$ of the total number of inputs $m$.
Since the original submission of this paper
and independently of our work, Feldman et al.\ \cite{FM+08} showed that for this
sort of model, an analogue of Theorem \ref{thm:main} holds
where the divide-and-conquer version uses strings of length at most
$k^2(m)$.
Their construction, like ours, takes exponential time in $k(m)$.

Finally, the functional composition view of FSAs (e.g., in the proof
of Theorem \ref{thm:compose}) also works for nondeterministic
automata
and probabilistic automata.
A result obtained by Feldman et al.\ \cite{FM+08} shows
that an analogue of Theorem \ref{thm:main} for
probabilistic automata is false,
while the nondeterministic version appears to be an open problem.

\bibliography{../huge}
\bibliographystyle{jca}

\typeout{Label(s) may have changed. Rerun}

\end{document}